\documentclass[conference]{IEEEtran}
\IEEEoverridecommandlockouts
\usepackage{cite}
\usepackage{amsmath,amssymb,amsfonts}
\usepackage{algorithmic}
\usepackage{graphicx}
\usepackage{textcomp}
\usepackage{xcolor}
\usepackage{multirow}
\usepackage{makecell}
\usepackage{url}
\usepackage{amsthm}
\theoremstyle{definition}
\newtheorem{dfn}{Definition}

\newtheorem{thm}[dfn]{Theorem}

\def\BibTeX{{\rm B\kern-.05em{\sc i\kern-.025em b}\kern-.08em
    T\kern-.1667em\lower.7ex\hbox{E}\kern-.125emX}}
    
\newcommand{\mcrot}[4]{\multicolumn{#1}{#2}{\rlap{\rotatebox{#3}{#4}~}}} 
\begin{document}

\title{\textsc{FiberPool}: Leveraging Multiple Blockchains \\for Decentralized Pooled Mining
}

\author{\IEEEauthorblockN{Akira Sakurai}
\IEEEauthorblockA{\textit{Kyoto University} \\
Kyoto, Japan}
\and
\IEEEauthorblockN{Kazuyuki Shudo}
\IEEEauthorblockA{\textit{Kyoto University} \\
Kyoto, Japan}
}

\maketitle

\begin{abstract}
The security of blockchain systems based on Proof of Work relies on mining. However, mining suffers from unstable revenue, prompting many miners to form cooperative mining pools. Most existing mining pools operate in a centralized manner, which undermines the decentralization principle of blockchain.

Distributed mining pools offer a practical solution to this problem. Well-known examples include \textsc{P2Pool} and \textsc{SmartPool}. However, \textsc{P2Pool} encounters scalability and security issues in its early stages. Similarly, \textsc{SmartPool} is not budget-balanced and imposes fees due to its heavy use of the smart contract.

In this research, we present a distributed mining pool named \textsc{FiberPool} to address these challenges. \textsc{FiberPool} integrates a smart contract on the main chain, a storage chain for sharing data necessary for share verification, and a child chain to reduce fees associated with using and withdrawing block rewards. We validate the mining fairness, budget balance, reward stability, and incentive compatibility of the payment scheme—\textsc{FiberPool} Proportional—adopted by \textsc{FiberPool}.
\end{abstract}


\section{Introduction}
Bitcoin~\cite{bitcoin} and other Proof of Work (PoW)–based blockchains derive their security from well-designed incentive mechanisms. Nodes participating in the blockchain network, known as miners, perform mining to receive block rewards, thereby stabilizing the system and processing transactions.

One issue associated with mining is the instability of mining revenue. The time interval between block generations for each miner follows an exponential distribution with parameter $\lambda = \alpha / T$, where $\alpha$ denotes the ratio of the miner's hashrate and $T$ represents the average block generation interval. Consequently, the expected time until receiving a block reward is $T/\alpha$, and the variance is $T^2/\alpha^2$. This implies that miners with lower hashrates experience greater instability in their mining revenue.

To mitigate the instability of mining revenue, most miners cooperate by pooling their computational resources to obtain block rewards more steadily. Such a system is known as a mining pool. Miners in a mining pool demonstrate their contributions by submitting partial proofs of work, called shares, and receive rewards determined by the pool’s payment scheme in accordance with their contributions.

However, most existing mining pools operate in a centralized manner, with the majority of the system's hashrate concentrated within these pools~\cite{IsBitcoinaDecentralizedCurrency}. This concentration undermines the overall decentralization of the system and leads to issues such as participation fees, centralization of transaction-processing nodes, and an increased risk of attacks, including Selfish Mining~\cite{majorityisnotenough} and double-spending attacks.

As a practical approach to the centralization problem of mining pools, distributed mining pools offer a promising solution. \textsc{P2Pool}~\cite{p2pool} was the first distributed mining pool proposed. It is managed by a side chain of shares, called a share chain. This blockchain of shares—easier to generate than blocks—has a shorter generation interval than the main chain and can distribute rewards to more miners.

However, \textsc{P2Pool} faces two main issues. The first is scalability: as a blockchain, the share chain has a limited capacity for generating shares. The second issue concerns the security of the share chain, which depends on the hashrate of the miners participating in \textsc{P2Pool}. In particular, the low security during the pool's early stage is problematic.

As a solution to the problems of \textsc{P2Pool}, \textsc{SmartPool}~\cite{SmartPool} was proposed. \textsc{SmartPool} is a distributed mining pool managed by a smart contract on the main chain, meaning its security relies on the security of that chain. Additionally, \textsc{SmartPool} uses probabilistic verification via a Merkle tree~\cite{merkletree} for share verification by the smart contract. This probabilistic approach allows a large number of shares to be processed relatively inexpensively.

However, \textsc{SmartPool} faces two main issues. The first is soaring fees caused by performing share verifications entirely on the main chain. The second issue is its adoption of the Pay-Per-Share (PPS) payment scheme. PPS distributes rewards for shares submitted at a predetermined rate, but it is not budget-balanced. This imbalance undermines system decentralization and effectively incurs fees charged to the administrator.

In this research, we propose a distributed mining pool named \textsc{FiberPool} that addresses the shortcomings of existing distributed mining pools (see Table~\ref{table:compare}). \textsc{FiberPool} is governed by three blockchains: a smart contract on the main chain, a storage chain, and a child chain. In \textsc{FiberPool}, data necessary for share verification is shared via the storage chain, and share verification is performed locally by each miner. This approach significantly reduces fees associated with share verification. Moreover, fees for using or withdrawing rewards are lowered by employing the child chain, a layer-2 technology~\cite{BlockchainScalingUsingRollupsAComprehensiveSurvey}\cite{plasma}.

Additionally, \textsc{FiberPool} introduces a novel payment scheme called \textsc{FiberPool} Proportional (FProportional). We evaluate this payment scheme in terms of mining fairness, budget balance, reward stability, and incentive compatibility.

\begin{table}[tb]
  \centering
  \caption{Comparison between existing mining pools and \textsc{FiberPool}.}
  \label{table:compare}
  \begin{tabular}{cccccc}
     & \mcrot{1}{c}{60}{Decentralization} & \mcrot{1}{c}{60}{Scalability} & \mcrot{1}{c}{60}{Early-Stage Security} & \mcrot{1}{c}{60}{Blockchain Fee} & \mcrot{1}{c}{60}{Budget Balance} \\ \hline \hline
     \multicolumn{1}{c||}{Centralized Pool} & $\times$ & $-$ & $-$ & $-$ & $-$ \\ 
     \multicolumn{1}{c||}{\textsc{P2Pool}} & $\circ$ & $\times$ & $\times$ & $\circ$ & $\circ$ \\
     \multicolumn{1}{c||}{\textsc{SmartPool}} & $\circ$ & $\circ$ & $\circ$ & $\times$ & $\times$ \\
     \multicolumn{1}{c||}{\textsc{FiberPool} (proposed)} & $\circ$ & $\circ$ & $\circ$ & $\circ$ & $\circ$ \\ \hline
  \end{tabular}
\end{table}

\section{Existing Mining Pools and Challenges}
\subsection{Mining Pools}
A mining pool is a system in which miners cooperate to generate blocks and distribute the block rewards to each miner according to their contributions. By using a mining pool, miners can receive rewards steadily. The contribution of each miner is measured by shares. A share refers to a PoW that is of lower difficulty than the block.

There are four ideal properties that a payment scheme should satisfy:
\begin{itemize}
    \item \textbf{Mining Fairness}: Rewards are given according to each miner’s contribution to the pool. The ratio of hashrate of each miner to the block reward rate should be equal.
    \item \textbf{Budget Balance}: The rewards earned by the pool are distributed to the miners without surplus or deficit.
    \item \textbf{Reward Stability}: Rewards are received in a stable manner.
    \item \textbf{Incentive Compatibility}: Honest mining is more economically advantageous than strategic mining.
\end{itemize}

As for incentive compatibility, various strategic mining methods can be considered, but this study considers three strategic mining tactics: pool hopping, cross-period mining strategies, and delaying the submission of shares and blocks. This research does not consider strategies including intentional forks of the main chain aimed at invalidating blocks of honest miners. Regarding incentive compatibility that takes intentional forks of the main chain into account, it is known that no existing payment scheme or mining pool can satisfy such compatibility in a way that maintains compatibility with the system \cite{FAW}. Suppressing such strategic mining is not the purpose of this study.

There are several payment schemes. Here, we focus on Proportional, Pay-Per-Share (PPS), and Pay-Per-Last-N-Shares (PPLNS) as examples. 

Proportional distributes rewards according to the ratio of shares submitted by each miner until the pool as a whole finds a block. Under this method, incentive compatibility does not hold \cite{rosenfeld2011analysisbitcoinpooledmining}. For instance, if the proportion of shares at the time of block generation is lower than the proportion of hashrate, it is more profitable to delay the publication of the block.

PPS is a payment scheme that exchanges shares for rewards at a predetermined rate. The problem with PPS is that it is not budget-balanced. If no shares are generated before a block is found, most of the block reward will not be distributed to the miners. Moreover, if shares are generated continuously without a block ever being found, the pool's funds may be depleted.

PPLNS distributes the block reward among the most recent N shares (including the block itself) submitted to the pool at the time of block generation.

\subsection{\textsc{P2Pool}}
\textsc{P2Pool} is a distributed mining pool managed by a share chain, a type of side chain \cite{p2pool}. Similar to how blocks work in a blockchain, shares in \textsc{P2Pool} embed information referencing another share, thereby forming a list structure in the same way as a blockchain. When a miner generates a share, they publish it to be shared across the entire network. Upon receiving a share, each miner verifies the PoW of that share and checks whether the coinbase transaction distributes the block reward among the past last N shares (where N is a value determined by the protocol) that include this share. After verification, miners update their own share chain. 

When a miner generates a block, the block reward is distributed equally among the past N shares, including the block itself (see Figure~\ref{P2Pooloverview}). This implies that the payment scheme used by \textsc{P2Pool} is PPLNS.

\begin{figure}[tb]
\begin{center}
\includegraphics[width=\linewidth]{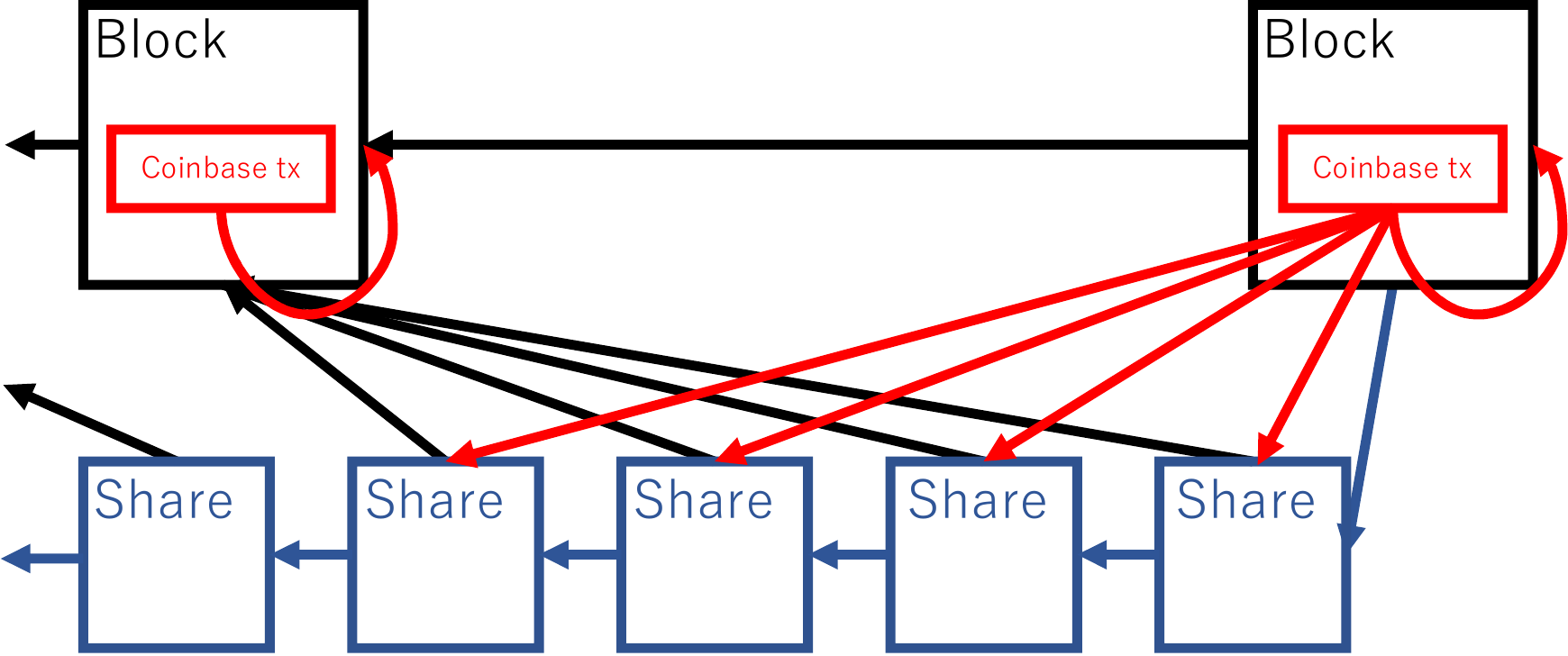}
\end{center}
\caption{Overview of \textsc{P2Pool}. Black arrows represent the reference structure of blocks, blue arrows represent the reference structure of shares, and red arrows indicate the distribution of rewards by the coinbase transaction. \textsc{P2Pool} is managed by a share chain with a shorter generation interval than the main chain. When a share on the share chain also functions as a block, the block reward is distributed over the past $N$ shares.}
\label{P2Pooloverview}
\end{figure}

There is nothing that forces miners participating in \textsc{P2Pool} to use a specific block template. For example, it is possible to generate a block that does not distribute rewards to the past N shares. However, such a block will not be considered a valid share by other miners. In other words, even if the PoW is valid as a share, it will not be subject to reward distribution. This is equivalent to mining without participating in \textsc{P2Pool}.

\textsc{P2Pool} faces the following two problems. First, the share chain is ultimately a blockchain, and thus it is subject to the well-known scalability problem of blockchains \cite{OnScalingDecentralizedBlockchains}. Indeed, recent research has significantly increased the number of transactions blockchains can handle compared to those early days \cite{Prism}\cite{phantomghostdag} but achieving blockchain processing capabilities that are both inexpensive and exceed network capacity remains challenging. In \textsc{P2Pool}, as the pool grows in size, the share generation difficulty increases, undermining reward stability. To mitigate this effect and stabilize rewards, it is necessary to increase $N$ and lower the difficulty of share generation. This corresponds to increasing the size of the coinbase and shortening the share block generation interval-both of which are reduced to the blockchain scalability problem.

The second problem is that since the share chain is a blockchain, its security depends on the total hashrate of \textsc{P2Pool} \cite{SmartPool}. This issue becomes especially problematic during the early stage when there are few miners participating in \textsc{P2Pool}.

In our proposed \textsc{FiberPool}, we tackle the scalability problem by integrating probabilistic share verification, similar to the method used in \textsc{SmartPool}. This strategy allows each miner to independently set the share generation difficulty, ensuring stable rewards regardless of the pool's size. Additionally, to address the low early-stage security issue inherent in \textsc{P2Pool}, \textsc{FiberPool} is built exclusively on blockchains that have already established sufficient security.

\subsection{\textsc{SmartPool}} \label{sectionsmartpool}
\textsc{SmartPool} was proposed as a decentralized mining pool to address the issues of \textsc{P2Pool} \cite{SmartPool}. \textsc{SmartPool} is managed by the smart contract. Since the smart contract is on the main chain, the security of the system relies on the main chain. In this sense, \textsc{SmartPool} is more secure than \textsc{P2Pool}. Additionally, \textsc{SmartPool} uses a Merkle tree for probabilistic verification of shares. By reducing the verification of a large number of shares to the verification of a few shares, \textsc{SmartPool} dramatically decreases computational costs and data size. Moreover, \textsc{SmartPool} allows the share generation difficulty to be set freely, enabling stable rewards.

The probabilistic verification of shares in \textsc{SmartPool} is also used in \textsc{FiberPool}, so we will describe it in detail here. First, each miner sets the share generation difficulty $D$ according to their computational power at will. Then, the validity of a share having a valid PoW can be verified by the following inequality:
\begin{align}
PoW(\text{share}) \leq D
\end{align}
Here, the function $PoW()$ refers to the Proof-of-Work algorithm. The miner embeds the difficulty $D$, the number of shares generated so far ($counter$), and their public key into the block template. This ensures a commitment to $D$, $counter$, and the public key. In addition, mining is performed by specifying the address of \textsc{SmartPool}'s smart contract as the coinbase.

By committing to the $counter$ and the public key during mining, the duplication and reuse of shares within a batch, within the same miner, or between different miners are prevented.

Miners send their shares to the smart contract when they wish to receive rewards, and then they receive the corresponding compensation. First, a Merkle tree is constructed from the list of shares arranged in order of $counter$. The tuple consisting of the Merkle root, the generation difficulty $D$, the number of shares, and the signature of the block generator is then submitted as a batch to the \textsc{SmartPool} smart contract. Next, the smart contract randomly selects one share from the list corresponding to the Merkle root (for instance, using the block hash value \cite{SmartPool} or RANDAO \cite{randao} as the random value), and requests the miner to submit it. The miner then submits the specified share, its Merkle proof, and the signature to the smart contract as proof. The smart contract verifies the submitted proof. Specifically, the verification is performed through the following steps:
\begin{enumerate}
   \item Verify that the position in the list determined by the Merkle proof and the share's $counter$ matches the position specified by the smart contract.
   \item Verify that the coinbase address specifies the \textsc{SmartPool} smart contract.
   \item Verify that the generation difficulty $D$ of the batch matches the share's generation difficulty $D$.
   \item Verify that the share has a PoW satisfying the generation difficulty $D$.
   \item Verify that the signature used to submit the batch and the signature sent to the smart contract for share verification were made using the private key corresponding to the public key included in the share.
\end{enumerate}
If all verifications succeed, the smart contract deems all shares in the batch valid and grants a reward corresponding to the total PoW held by the batch ($N/D$). If any verification fails along the way, the smart contract considers all shares in the batch as invalid and does not issue any reward. This mechanism implies that the payment scheme adopted by \textsc{SmartPool} is PPS (Pay-Per-Share).

The effectiveness of probabilistic verification of shares in \textsc{SmartPool} is based on the following economic insight. Let $R$ be the expected reward, $B$ be the reward received when share verification is successful, and $f$ be the fraction of invalid shares in the batch. Then the following equation holds:
\begin{align}
R &= B (1 - f) + 0 \cdot f \\
  &= B (1 - f)
\end{align}
This demonstrates that the expected reward remains unchanged whether or not invalid shares are included.

\textsc{SmartPool} faces the following two issues. First, there are fees due to the use of smart contracts on the main chain. This is essentially a scalability problem in \textsc{SmartPool}. Fees incurred from the direct use of smart contracts on the main chain are a universal issue in blockchain beyond just decentralized mining pools. To address this problem, there is active research and development in areas such as Layer 2 solutions and sharding \cite{lightningnetwork}\cite{plasma}\cite{Arbitrum}\cite{ShardingProtocol}. In addition, year by year, the complexity of the PoW verification required is increasing as a countermeasure against ASICs \cite{cuckooCycle}\cite{ProgPoW}\cite{etchash}\cite{randomx}\cite{scriptbitcoincash}, which brings about further surges in fees.

The second issue is that \textsc{SmartPool} adopts PPS as its payment scheme for submitted shares. A problem with PPS is that it is not budget-balanced, meaning that the block rewards earned by the pool and the rewards distributed to miners do not match perfectly. In fact, if enough blocks are not generated, PPS can lead to system bankruptcy. Such risks require appropriate countermeasures, such as having the smart contract manager reserve a certain amount of funds in the smart contract in advance. The presence of such a manager undermines the decentralization of the system. Furthermore, if rewards are not sufficiently distributed, the undistributed rewards virtually become fees for the miners.

In our proposed \textsc{FiberPool}, share verification does not involve using smart contracts on the main chain as was done in \textsc{SmartPool}. Instead, the data necessary for share verification is shared via a blockchain dedicated to data storage, and each miner performs share verification locally. This approach reduces the costs associated with share verification. Additionally, while \textsc{SmartPool} faces budget imbalance issues by pooling block rewards in smart contracts and distributing them based on submitted shares, \textsc{FiberPool} addresses this problem by directly allocating block rewards to miners according to their mining contributions.

\section{\textsc{FiberPool}}
We propose a decentralized mining pool, \textsc{FiberPool}, that addresses the challenges faced by existing solutions. \textsc{FiberPool} is managed by a smart contract on the main chain, a storage chain, and a child chain. Unlike \textsc{P2Pool}, which becomes insecure when fewer miners participate, \textsc{FiberPool} remains secure because the security of its three blockchain components does not depend on the computational resources within \textsc{FiberPool}.

In \textsc{FiberPool}, the data required for share verification is shared via the storage chain, which is well-suited for data storage. Furthermore, each miner performs share verification locally, thereby reducing fees associated with share verification. Additionally, the use of the child chain lowers fees related to reward withdrawals and expenditures.

Miners in \textsc{FiberPool} receive stable rewards regardless of pool size by employing probabilistic share verification, which allows each miner to freely set the share generation difficulty. Finally, \textsc{FiberPool} preserves a balanced budget by creating block templates that reflect historical mining contributions and distributing rewards directly to miners.

Considering an ideal mining pool, it would be preferable to distribute block rewards precisely based on each miner's contribution at the moment the pool successfully generates a block. However, this approach would require continuous knowledge of each miner's hashrate distribution, leading to significant communication costs and impracticality. To address this, \textsc{FiberPool} introduces the concept of a \emph{period}—a fixed time interval—and distributes rewards based on contributions within each period.

To accurately account for each miner's contribution per period, sufficient time must be allocated to verify the distribution of miners' hashrates. This means that mining activities and their verification do not align perfectly in time. Considering this misalignment raises an incentive compatibility issue related to the timing of reward distribution and mining strategies. For instance, if the verification of a period's hashrate distribution occurs after rewards for that period are finalized, miners might be economically motivated to concentrate their efforts in periods with higher rewards. To prevent this, \textsc{FiberPool} adjusts the process so that rewards for a period are determined only after the verification of that period's hashrate distribution is complete.

\subsection{Period}
We introduce a fixed time interval called a \emph{period}. The length of a period must be long enough to allow verification of all miners' shares—for example, the number of main chain blocks produced in one week. In \textsc{FiberPool}, each miner updates their block template at the start of every period and conducts mining accordingly. After completing mining for a given period, miners share the data necessary for share verification via the storage chain. The end of each period acts as a preparation phase, termed the \emph{Prepare} phase, during which each miner constructs block templates for the next period.

\begin{figure}[tb]
\begin{center}
\includegraphics[width=\linewidth]{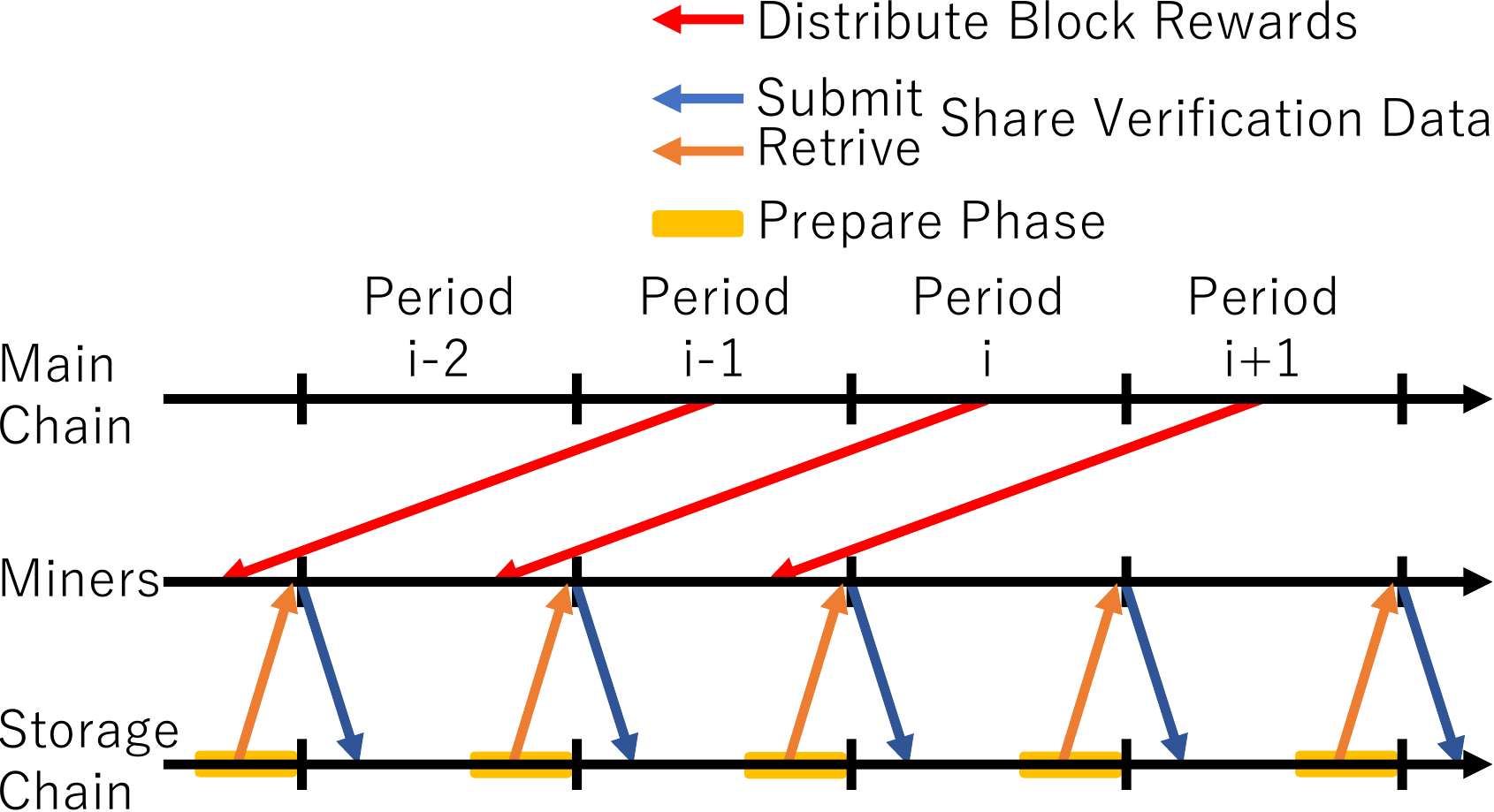}
\end{center}
\caption{Overview of mining in \textsc{FiberPool}. The red arrows indicate the distribution of block rewards: rewards from blocks generated in period $i$ are distributed to miners who mined in period $i-2$. The blue and orange arrows represent the submission and retrieving of data for share verification. After completing mining in period $i$, miners submit share verification data to the storage chain (blue arrows). Just before the start of period $i+2$, the system enters a \emph{Prepare} phase (yellow box), during which each miner retrieves the share verification data from the storage chain and calculates the total PoW and its distribution for period $i$ (orange arrows). By using the storage chain, the total PoW and its distribution for period $i$ are shared among miners before period $i+2$ begins, and in period $i+2$, miners commit to these values as they proceed with mining.}
\label{proposedmethod}
\end{figure}

Figure~\ref{proposedmethod} illustrates the interactions between the main chain, the storage chain, and the miners participating in \textsc{FiberPool}. In period $i$, the shares each miner generates based on their block template are sent in a batch to the storage chain before the \emph{Prepare} phase of period $i+2$ begins, where they are verified. The share verification data on the storage chain are then shared among all miners, so the total PoW and its distribution for period $i$ become known to all participants. At the beginning of period $i+2$, miners commit to the total PoW and its distribution and use this information to construct new block templates for mining.

\subsection{Block Template} \label{blocktemplatesection}
We describe the block template for period $i$. At the start of period $i$, the total PoW and its distribution from period $i-2$ have been shared among all miners. During period $i$, mining is conducted by committing to the total PoW and its distribution from period $i-2$.

\begin{figure}[tb]
\begin{center}
\includegraphics[width=\linewidth]{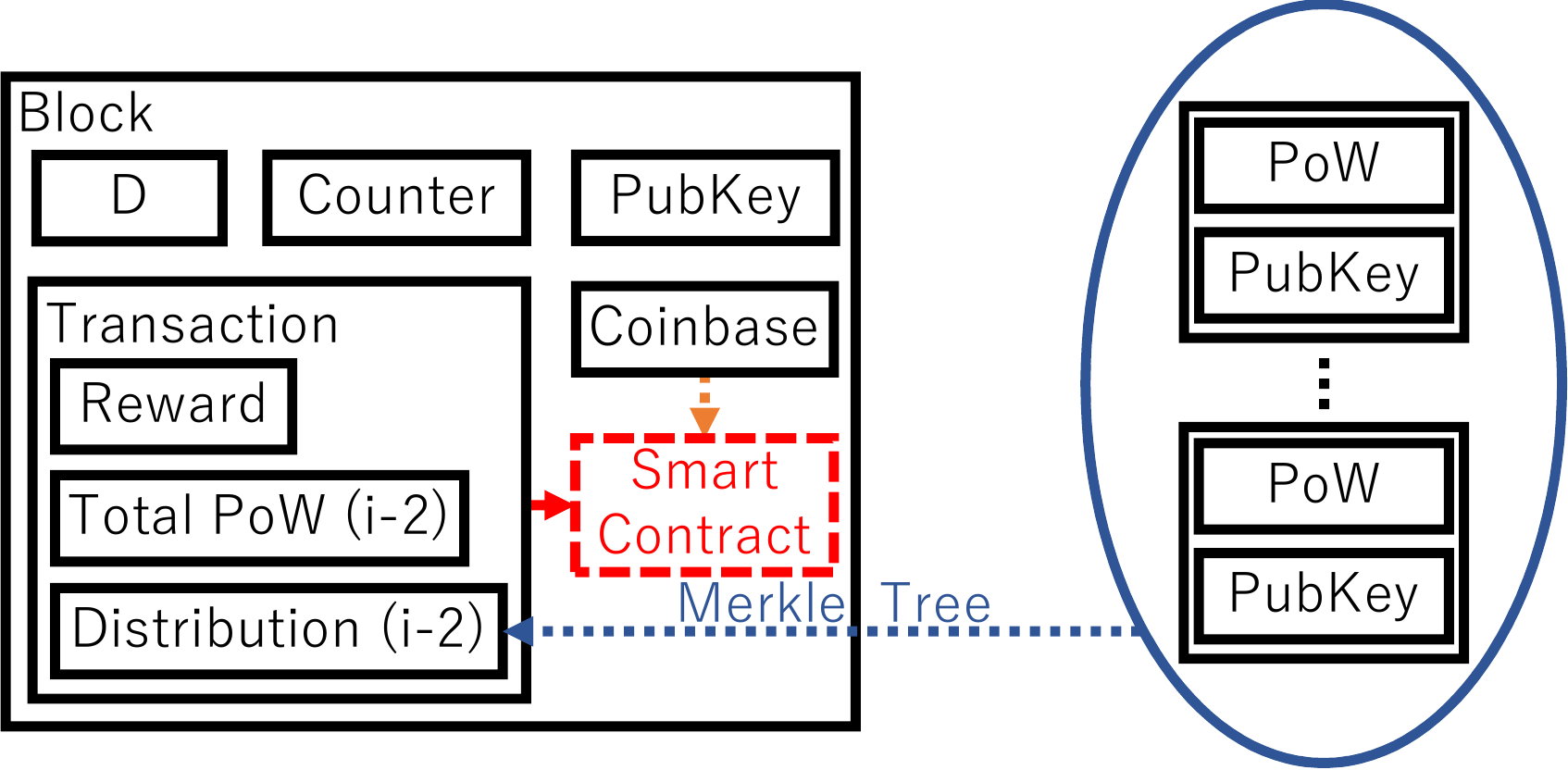}
\end{center}
\caption{Block template for period $i$ in \textsc{FiberPool}. Before period $i$ begins, the total PoW and its distribution from period $i-2$ are shared through the storage chain. During period $i$, mining involves committing to the total PoW and its distribution from period $i-2$.}
\label{blocktemplate}
\end{figure}

Figure~\ref{blocktemplate} presents a concrete block template for period $i$. First, the template's coinbase directs the reward to the smart contract address in \textsc{FiberPool}, ensuring that the entire block reward is sent to the smart contract. The template also includes a transaction linking the block reward to the total PoW and its distribution from period $i-2$. This transaction commits to the total PoW and its distribution from period $i-2$.

Additionally, the template contains information needed for the probabilistic verification of shares. Specifically, it embeds the share generation difficulty $D$, the number of shares generated in period $i$ ($counter$), and the block producer’s public key $pubkey$. These elements can be embedded, for example, by including them in the transaction to the smart contract.

In \textsc{FiberPool}, block rewards obtained in period $i$ are linked to the total PoW and its distribution from period $i-2$ by transactions. Each such transaction takes as arguments the block reward along with the total PoW and its distribution from period $i-2$. 

However, during transaction processing, balance changes resulting from the coinbase have not yet occurred, so the transaction cannot immediately transfer the block reward. Consequently, it is not feasible within that transaction to verify whether the block reward was transferred appropriately without surplus or deficit.

To overcome this limitation, \textsc{FiberPool} delegates verification to the next user of the smart contract. Specifically, the subsequent user checks whether the sum of the smart contract’s current balance and the total amount withdrawn thus far (i.e., the cumulative block rewards transferred) is at least equal to the total block rewards linked by all previous validated transactions and the unverified transaction. If this condition is met, the transaction is considered legitimate and the linking is completed. If not, the transaction is deemed illegitimate and the linking is invalidated.

This verification process is performed before the user utilizes the smart contract. By employing this approach, the smart contract prevents the improper use of block rewards that have not yet been accurately linked to the correct total PoW and its distributions.

As with \textsc{P2Pool}, \textsc{FiberPool} does not enforce miners to commit to specific values when constructing a block template. For instance, a miner could commit to the entire total PoW and its distribution so that the entire block reward would be sent to himself. However, if such a block template is constructed, then—even if the miner generates shares—those shares will be considered invalid during share verification. Consequently, they will not be eligible for block reward distribution in the following period. This means that the miner is not participating in \textsc{FiberPool} from the start.

Therefore, anyone participating in \textsc{FiberPool} must construct a block template in period $i$ by committing to the total PoW and its distribution from period $i-2$, as described above. The same requirement applies to other values such as $D$ and $counter$.

\subsection{Verification of Shares}
Each miner verifies the shares generated in period $i$ during the $Prepare$ phase of period $i+2$ through the storage chain. The purpose of the storage chain is to ensure that at the beginning of period $i+2$, all miners share the total PoW and its distribution from period $i-1$ (see Section~\ref{sectionstoragechain} for storage chain candidates).

The share verification process in \textsc{FiberPool} is fundamentally similar to the probabilistic share verification used in \textsc{SmartPool}. A miner who has completed mining in period $i$ submits a batch of generated shares to the storage chain. Here, the batch consists of the Merkle tree root of the shares generated during period $i$, the share generation difficulty $D$, and the miner's signature. Subsequently, a proof based on a random value (e.g. the hash value of the next block in the storage chain) is submitted to the storage chain. The proof comprises the share designated by the random value, the Merkle proof of the share, and the miner's signature. Unlike the probabilistic share verification in \textsc{SmartPool}, \textsc{FiberPool} requires all miners to submit share verification data to the storage chain for each period.

Before period $i+2$ begins, all miners must be prepared for mining in that period. Specifically, they need to share the total PoW and its distribution from period $i$ among all participants. To ensure this, a dedicated $Prepare$ phase is established before period $i+2$. During this phase, no share verification data for period $i$ can be submitted. This restriction ensures that only information with sufficient finality is retrieved from the storage chain, enabling the accurate calculation of the total PoW and its distribution for period $i$.

As an example, consider \textsc{FiberPool}, where the period is determined by the block height on the main chain. The $Prepare$ phase for period $i$ can be established as follows: first, use the main chain's block height to designate a specific block that marks the start of the main chain's $Prepare$ phase for each period. Next, locate the first block on the storage chain whose timestamp exceeds that of the designated main chain block. This storage chain block then marks the beginning of the $Prepare$ phase for period $i$ on the storage chain.

Specifically, each miner performs locally the following verification on the share verification data generated during period $i$, which was retrieved from the storage chain during the $Prepare$ phase:
\begin{enumerate}
   \item Verify that the verification data was submitted to the storage chain before the $Prepare$ phase of period $i+2$.
   \item Verify that the position in the list determined by the Merkle proof matches both the share's $counter$ and the list position specified by the random value.
   \item Verify that the block rewards have been sent to the smart contract.
   \item Verify that the transfer to the smart contract is linked to the total PoW and its distribution from period $i-2$.
   \item Verify that the batch's generation difficulty $D$ matches the share's generation difficulty $D$.
   \item Verify that the share has a PoW satisfying the generation difficulty $D$.
   \item Verify that the signature used to submit the batch and the signature sent to the storage chain for share verification were generated using the private key corresponding to the public key contained in the share.
\end{enumerate}
If all verifications succeed, all shares in the batch are considered valid. Each miner's PoW is given by $N/D$ where $N$ is the number of shares in the batch submitted by the miner. Otherwise, all shares are regarded as invalid. 

The major difference from the probabilistic verification of shares in \textsc{SmartPool} is that in \textsc{FiberPool}, a transaction to the associated smart contract is required, and it must be verified that the transfer to the smart contract is linked to the total PoW and its distribution from period $i-2$.

\begin{figure}[tb]
\begin{center}
\includegraphics[width=\linewidth]{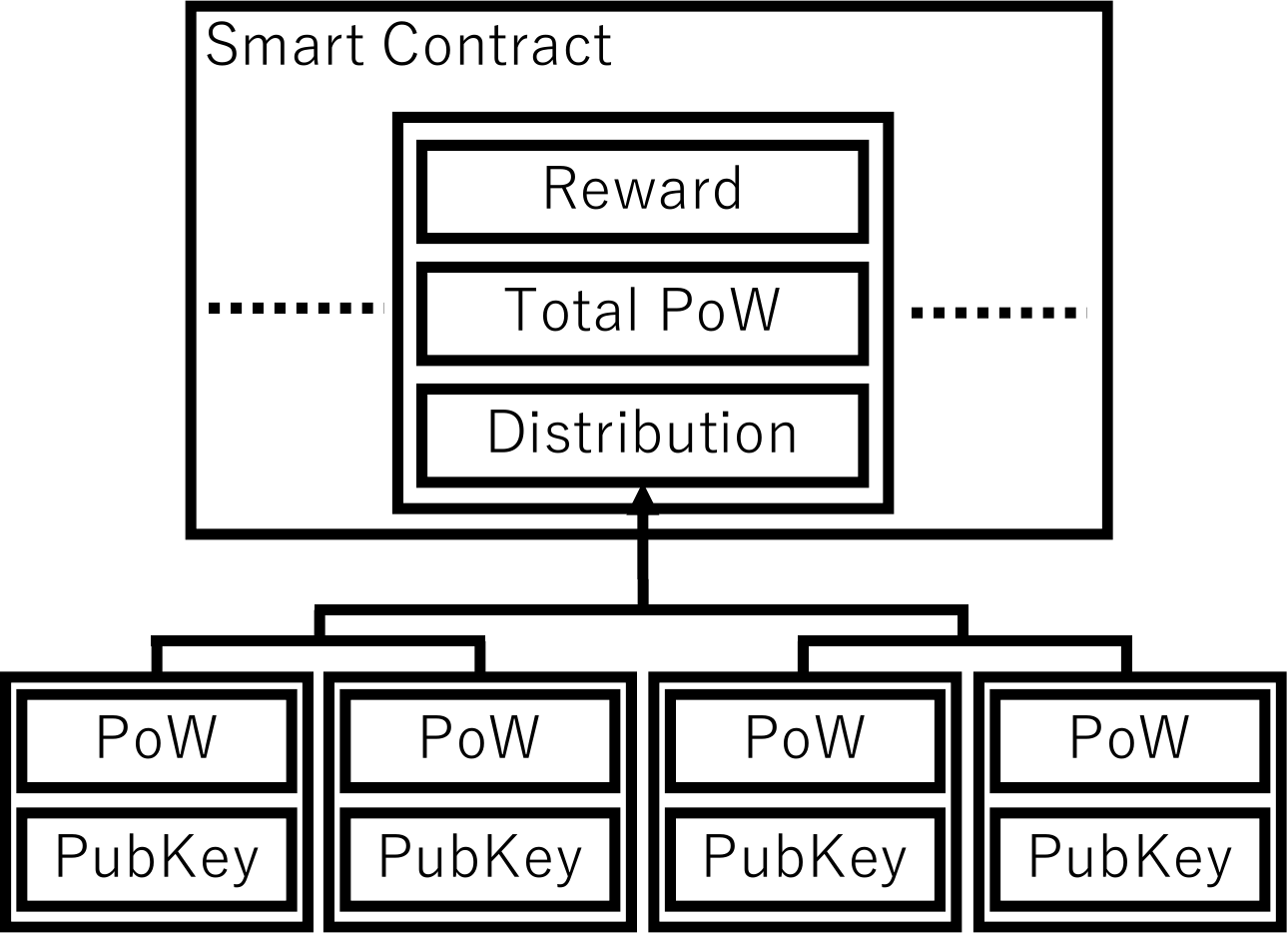}
\end{center}
\caption{Smart contract in \textsc{FiberPool}. Block rewards earned in period $i+2$ are linked to the total PoW and its distribution from period $i$.}
\label{smartcontract}
\end{figure}

\subsection{Using and Withdrawing Block Rewards}
In \textsc{FiberPool}, fees associated with using and withdrawing block rewards are reduced by employing Layer 2 technologies based on a child chain, such as Rollup or Plasma \cite{plasma}\cite{Arbitrum}\cite{BlockchainScalingUsingRollupsAComprehensiveSurvey}.

By treating the transfer of block rewards in \textsc{FiberPool} as a currency deposit to the child chain, the child chain can be integrated into \textsc{FiberPool} with minimal changes. The primary modification involves adding a transfer of block rewards linked to the total PoW and its distribution as a deposit mechanism (see Section~\ref{blocktemplatesection} and Figure~\ref{smartcontract}). Additionally, when using such a deposit, one pair of PoW and public key from the corresponding total PoW and its distribution must be specified, along with the associated Merkle proof and signature. 

In the case of Plasma, which utilizes a challenge/assertion model, the priority for exiting block rewards that have not been transacted within the child chain must be set to an additional three periods in the past. This adjustment is necessary because it takes up to three periods to fully cease participation in \textsc{FiberPool}, thereby requiring an increase in the exit priority of block rewards.

By utilizing a child chain, block rewards can be used inexpensively on the child chain without withdrawal. Additionally, block rewards on the child chain can be aggregated into a single transaction, and the single transaction can be used to withdraw the accumulated block rewards all at once.

\section{Evaluation}
We evaluated \textsc{FiberPool} by examining key desirable properties of a mining pool within the system. These properties include mining fairness, budget balance, reward stability, and incentive compatibility.

Assume three miners: miner $m_1$, who freely chooses whether to participate in \textsc{FiberPool} or not; miner $m_2$, who participates in \textsc{FiberPool}; and miner $m_3$, who does not participate in \textsc{FiberPool}. Let the fraction of the hashrate held by each miner be $\alpha$ for $m_1$, $\beta$ for $m_2$, and $\gamma$ for $m_3$. Then, $\alpha + \beta + \gamma = 1$. We assume that the shares submitted and verified by each miner in each period accurately reflect the miner's hashrate. This assumption is based on the fact that each miner can freely set the share generation difficulty $D$, resulting in a negligible coefficient of variation in the number of generated shares.

\subsection{Mining Fairness and Budget Balance}
Mining fairness means that rewards are distributed proportional to mining contributions. Formally, the ratio of a miner's hashrate to the block reward rate should be equal. Under \textsc{FiberPool}, mining fairness holds in the following sense.

\begin{thm} \label{thmfairness} Consider period $i$. Suppose miner $m_1$ has continuously participated in \textsc{FiberPool} since before period $i-2$. Then, each time a miner participating in \textsc{FiberPool} generates a block in period $i$, $m_1$ receives a fraction of the block reward equal to $\alpha/(\alpha + \beta)$. \end{thm}

\begin{proof} Miner $m_1$ and miner $m_2$ both mine during period $i-2$ and submit their share verification data to the storage chain before the $Prepare$ phase of period $i$. At the beginning of period $i$, both $m_1$ and $m_2$ share the total PoW and its distribution from period $i-2$ and commit to this data while mining. At this time, the fraction of total hashrate contributed by $m_1$ in period $i-2$ is $\alpha/(\alpha + \beta)$. Therefore, $m_1$ receives $\alpha/(\alpha + \beta)$ of the block rewards generated by \textsc{FiberPool} in period $i$. \end{proof}

Next, we demonstrate budget balance in \textsc{FiberPool}. Budget balance means that rewards are distributed to miners without surplus or deficit. Using an argument similar to the proof of Theorem~\ref{thmfairness}, one can show that miner $m_2$ receives a fraction of $\beta/(\alpha + \beta)$ of the block reward. Therefore, the total rewards for $m_1$ and $m_2$ equal the block reward. This verifies that the payment scheme in \textsc{FiberPool} is budget-balanced unlike \textsc{SmartPool}.

\subsection{Reward Stability}
The stability of rewards in a budget-balanced mining pool can be divided into two aspects: the overall stability of the pool's rewards and the stability of each individual miner's reward within the pool. The overall reward stability of the mining pool depends solely on the fraction of the total hashrate it controls, a topic already analyzed in detail in existing research \cite{rosenfeld2011analysisbitcoinpooledmining}. Therefore, we focus on the second factor: the stability of each miner's reward within the pool. This stability is influenced primarily by the miner's hashrate fraction and the payment scheme used.

To demonstrate the effectiveness of reward stability in \textsc{FiberPool}, we compare it with the Pay-Per-Last-N-Shares (PPLNS) payment scheme. Let the fixed block reward be $B$.

\begin{thm}\label{variance} Consider period $i$. Suppose miner $m_1$ has continuously participated in \textsc{FiberPool} since before period $i-2$. When the \textsc{FiberPool} mining pool generates a block in period $i$, the variance of $m_1$'s reward is 0. In contrast, under PPLNS, the variance of $m_1$'s reward is given by $p(1-p)B^2/N$, where $p = \alpha/(\alpha + \beta)$, and $N$ is the number of shares eligible for block reward distribution. \end{thm}

\begin{proof}
From Theorem~\ref{thmfairness}, miner $m_1$'s reward per block in \textsc{FiberPool} is fixed at $B\alpha/(\alpha + \beta)$, implying that the variance of $m_1$'s reward is 0. 

In contrast, under PPLNS, consider the shares among which the block reward is distributed at the time of block generation. The number of these shares belonging to $m_1$ follows a binomial distribution with success probability $p = \alpha/(\alpha + \beta)$. The variance of a binomial distribution with parameters $N$ (number of trials) and $p$ is $Np(1-p)$. Given that each share corresponds to a reward of $B/N$, the variance of $m_1$'s reward under PPLNS becomes
\begin{align}
Np(1-p)\left(\frac{B}{N}\right)^2 = p(1-p)\frac{B^2}{N}.
\end{align}
Here, $B/N$ denotes the reward allocated per share. 
\end{proof}

Additionally, in \textsc{FiberPool}, each miner can freely set the share generation difficulty. This means that by appropriately setting the difficulty, a miner can generate shares that accurately reflect their contribution through mining. Moreover, once a share is generated, it is guaranteed to be subject to reward distribution after undergoing the verification process within \textsc{FiberPool}. In other words, at the time of participation in \textsc{FiberPool}, a miner is assured of receiving rewards proportional to the ratio of its mining contribution. 

On the other hand, in PPLNS, a miner must generate shares at a difficulty level predetermined by the system or protocol to be eligible for rewards. The above discussion, together with Theorem~\ref{variance}, demonstrates that \textsc{FiberPool} offers higher reward stability compared to mining pools adopting PPLNS like \textsc{P2Pool}.

\subsection{Incentive Compatibility} \label{incentiveCompatibility}
Incentive compatibility ensures that each miner's most profitable strategy is to mine honestly. In this study, we assess the incentive compatibility of \textsc{FiberPool} by examining three potential mining strategies: pool hopping, cross-period mining strategies, and delaying the submission of shares or blocks.

\subsubsection{Pool Hopping}
\textsc{FiberPool} is resistant to pool hopping in the following sense.

\begin{thm}\label{poolhopping} 
Fix the total block reward for all miners at $R$ per period. Consider period $i$. Suppose that miner $m_1$ participates in \textsc{FiberPool} for $N-2$ consecutive periods starting from period $i$, and then exits \textsc{FiberPool} for the final two periods to mine independently. In this scenario, compared to not participating in \textsc{FiberPool}, $m_1$ suffers a loss of $\frac{2R\alpha^2}{\alpha + \beta}$. 
\end{thm}

\begin{proof}
During the first 2 periods starting from period $i$, $m_1$ receives no rewards. This is because all block rewards from blocks generated by $m_1$ are distributed to miners who participated in \textsc{FiberPool} before period $i-1$. In the subsequent $N-4$ periods, $m_1$ earns $R\alpha$ per period. This is because the total reward of \textsc{FiberPool} is $R(\alpha + \beta)$, and $m_1$ receives a fraction $\alpha / (\alpha + \beta)$. In the final two periods, the reward per period becomes
\begin{align}
R\Bigg(\alpha + \frac{\alpha\beta}{\alpha + \beta}\Bigg).
\end{align}
First, the block reward from blocks generated by $m_1$ is $R\alpha$ per period. Additionally, from the block rewards $R\beta$ generated by \textsc{FiberPool}, $m_1$ receives
\begin{align}
R\beta \cdot \frac{\alpha}{\alpha + \beta}
\end{align} 
per period.

If $m_1$ does not participate in \textsc{FiberPool}, the total reward would be $NR\alpha$. Therefore, by participating in \textsc{FiberPool}, the loss compared to not participating is
\begin{align}
    &\,NR\alpha - \Big(2 \cdot 0 + (N-4) \cdot R\alpha + 2 \cdot R\Big(\alpha + \frac{\alpha\beta}{\alpha + \beta}\Big)\Big) \\
    &= \frac{2R\alpha^2}{\alpha + \beta}.
\end{align}
\end{proof}

Theorem~\ref{poolhopping} demonstrates that if a miner engages in pool hopping with \textsc{FiberPool}, they incur a loss of $\frac{2R\alpha^2}{\alpha + \beta}$ each time they repeatedly enter and exit the pool. Consequently, miners have no incentive to adopt pool hopping strategies when participating in \textsc{FiberPool}.

As a potential drawback, one might worry that the incentive to participate in \textsc{FiberPool} is minimal. However, we believe this concern is minor. Mining in a PoW blockchain is a zero-sum game: one miner's loss is another's gain. The loss incurred by participating in \textsc{FiberPool} is offset by the benefits of continued participation in \textsc{FiberPool} for honest miners.

\subsubsection{Cross-period Strategic Mining}
In \textsc{FiberPool}, cross-period mining is possible. Until the $Prepare$ phase of period $i+1$ begins, it is possible to submit share verification data from period $i-1$ to the storage chain. This capability allows for the generation of valid shares for period $i-1$ during period $i$. We will now show that incorporating such cross-period mining into an otherwise honest strategy to earn higher profits is difficult.

Since the interval between periods is determined by the number of blocks in the main chain, we assume that the total block rewards of \textsc{FiberPool} per period remain constant regardless of the period. In addition, we assume that each miner's PoW per period is also constant.

Under these assumptions, the following theorem shows that cross-period mining does not yield higher profits than honest mining alone.

\begin{thm}\label{intermining}
Incorporating cross-period mining into an honest mining strategy is not more profitable than honest mining alone.
\end{thm}

\begin{proof}
Fix the total \textsc{FiberPool} reward per period at $R$. In fact, the more $m_1$ engages in cross-period mining, the smaller the \textsc{FiberPool} block reward becomes. Thus, this assumption favors cross-period mining.

Let $P$ be the total amount of PoW in the entire blockchain system per period. Although cross-period mining in \textsc{FiberPool} is executed in period $i$ for period $i-1$, we extend this scenario: assume that $m_1$ can freely allocate its PoW among each period over $N$ consecutive periods.

Maximizing $m_1$'s reward can then be reduced to the following optimization problem:
\begin{align}
\text{maximize} \quad & \sum\limits_{i=1}^N \frac{R x_i}{P(1-\alpha) + x_i}, \\
\text{subject to} \quad 
& \sum\limits_{i=1}^N x_i = NP\alpha, \label{lc1} \\
& x_i \geq 0 \quad \text{for all } i.
\end{align}
We solve this problem using the method of Lagrange multipliers. Considering the constraint (\ref{lc1}), define
\begin{align}
L(\mathbf{x}) = \sum\limits_{i=1}^N \frac{R x_i}{P(1-\alpha) + x_i} + \lambda \Big(NP\alpha - \sum\limits_{i=1}^N x_i\Big).
\end{align}
Next, take the partial derivative of $L$ with respect to each $x_i$:
\begin{align}
\frac{\partial L}{\partial x_i} = \frac{R \cdot P(1-\alpha)}{[P(1-\alpha) + x_i]^2} - \lambda.
\end{align}
Setting $\frac{\partial L}{\partial x_i} = 0$ yields
\begin{align}
\frac{R \cdot P(1-\alpha)}{[P(1-\alpha) + x_i]^2} = \lambda.
\end{align}
This implies that $x_i$ is the same for all $i$. From constraint (\ref{lc1}), we have
\begin{align}
x_i = P\alpha \quad \text{for all } i.
\end{align}
This result indicates that the optimal strategy of incorporating cross-period mining into honest mining is actually honest mining without cross-period mining.
\end{proof}

\textbf{Assumption of Constant Total Block Rewards per Period.}
When engaging in strategic mining that incorporates cross-period mining, an important consideration for miner $m_1$ is the variation of total block rewards in each period. In theory, by concentrating mining efforts on periods with higher rewards, one might earn more. However, in \textsc{FiberPool}, adjusting strategies based on differing block rewards is challenging. This difficulty arises because the block rewards actually distributed for mining activities in period $i$ originate from blocks in period $i+2$. As a result, shares for period $i$ must be submitted before the \textit{Prepare} phase of period $i+2$ begins.

\subsubsection{Mining Strategy of Delaying Submission of Shares or Blocks}
In a typical mining pool, shares and blocks are submitted immediately upon generation. However, by delaying these submissions, a miner might increase their rewards under certain payment schemes. For example, in a pool using a proportional payment scheme, delaying block submission until after generating additional shares could increase a miner’s reward.

Such timing-based strategies do not offer an advantage in \textsc{FiberPool}. In this system, share submission occurs before the \textit{Prepare} phase, and the actual reward distribution for these shares takes place in period $i+2$. Therefore, the timing of share submission does not affect a miner’s proportion of PoW or their rewards.

The same reasoning applies to blocks. In period $i$, regardless of when a miner publishes a block after its generation, the timing does not alter the total PoW contribution attributed to each miner in that period. Consequently, delaying block or share submissions does not increase rewards in \textsc{FiberPool}.

\section{Discussion}
\subsection{Related Work} Mining pools aim to stabilize mining revenue and are a common feature in most blockchain systems. The issues with mining pools fall into two main categories. The first category concerns the centralization of the entire system due to the concentration of mining power in a few centralized pools. The second category involves new attack risks that arise with the introduction of mining pools.

A well-known attack associated with mining pools is the Block Withholding Attack~\cite{rosenfeld2011analysisbitcoinpooledmining}\cite{OnSubversiveMinerStrategiesandBlockWithholdingAttack}\cite{miningdilemma}\cite{BitcoinBlockWithholdingAttackAnalysisandMitigation}\cite{FAW}. While Bag et al. proposed a method to prevent Block Withholding Attacks by modifying the PoW algorithm~\cite{BitcoinBlockWithholdingAttackAnalysisandMitigation}, a solution that remains compatible with existing systems has yet to be proposed.

Among the issues with mining pools, one promising approach to mitigate centralization is the development of decentralized mining pools. An early example is \textsc{P2Pool}\cite{p2pool}, which is managed by a share chain. There are also proposals that extend the payment scheme of \textsc{P2Pool} in a more general way\cite{Blockchain-basedDecentralizedRewardSharing}. However, \textsc{P2Pool} faces two key challenges: its security depends on the computational power of the share chain, and it is subject to the inherent scalability issues of blockchains. Solutions to the blockchain scalability problem include Prism\cite{Prism} and Phantom\cite{phantomghostdag}, and there have been attempts to apply these to \textsc{P2Pool}\cite{braidpool}. However, these attempts ultimately run into issues with the network resources and transaction processing capacity dilemma. 

To address these challenges, a decentralized mining pool managed by a smart contract, called \textsc{SmartPool}\cite{SmartPool}, was proposed. By leveraging smart contracts, \textsc{SmartPool} aligns the system’s security with that of the main chain. Moreover, it introduces probabilistic share verification, which significantly reduces the required network resources, data volume, and computation for share verification, thereby enabling stable reward distribution. However, \textsc{SmartPool} encounters issues such as fees incurred from using smart contracts on the main chain and a payment scheme that is not budget-balanced.

Other decentralized mining pool proposals include PoolParty\cite{PoolParty} and the approach by Papathanasiou et al.\cite{pooldecentralized}. PoolParty distributes rewards using payment channels such as the Lightning Network\cite{lightningnetwork}. A notable drawback of PoolParty is that it lacks detailed explanations on how these payment channels are utilized for reward distribution, leaving the method unclear. Additionally, PoolParty requires miners to provide a deposit equivalent to the block reward as a safeguard against fraud. This requirement substantially limits miner participation and undermines system decentralization.

Papathanasiou et al.'s decentralized mining pool reproduces the mechanisms of existing pools using smart contracts in a decentralized manner. However, this approach faces two significant issues: scalability problems similar to those of \textsc{P2Pool}, and fee challenges caused by share verification on smart contracts.

Another approach to preventing overall system centralization by mining pools is to change the PoW algorithm\cite{NonoutsourceableScratch-Off}\cite{disincentivizeLargeBitcoinMiningPools}. By altering the PoW algorithm, the incentive to form mining pools is reduced. However, this approach is not compatible with existing systems and presents greater implementation challenges compared to the development of decentralized mining pools.

\subsection{storage chain} \label{sectionstoragechain}
\textbf{Candidates for the storage chain.} In \textsc{FiberPool}, a storage chain is required to share verification data for shares. The primary requirement for the storage chain is to achieve consensus in a decentralized manner. Potential candidates include blockchains specialized in data storage, such as Filecoin\cite{filecoin} or Arweave\cite{arweave}, as well as a share chain similar to the one used in \textsc{P2Pool}. Using a blockchain specialized in data storage is straightforward and provides robust safety, especially in the early stages when insufficient miners are involved in \textsc{FiberPool}. On the other hand, a share chain could enable the sharing of share verification data at a lower cost. However, designing effective incentives for sharing verification data presents a challenge when employing a share chain as the storage chain. Additionally, the low security in the early stages of a share chain poses further difficulties.

\textbf{Reducing the cost of sharing share verification data.} In \textsc{FiberPool}, fees are incurred when sharing share verification data through the storage chain. One approach to reduce these fees is to aggregate verification data from multiple miners into a single transaction and then submit the aggregated verification data to the storage chain. This requires a single entity responsible for aggregation and submission. However, the presence of such an entity does not compromise the decentralization of \textsc{FiberPool}. This is because alternative submission methods that bypass censorship by this entity remain available. Even in the event of censorship, there is ample time to detect and respond, ensuring the submission process can proceed without hindrance.

\subsection{Impact of \textsc{FiberPool} from the View Point of Mining Fairness}
As established by Theorem~\ref{thmfairness}, \textsc{FiberPool} guarantees strong mining fairness. One clear advantage is the reduced risk of Selfish Mining and similar attacks. From a fairness standpoint, Selfish Mining can be seen as an attack that disrupts fair mining. Owing to \textsc{FiberPool}’s strong mining fairness, all miners receive rewards proportional to their hashrate, irrespective of an adversary’s strategic choices.

On the other hand, a disadvantage of strong mining fairness is that it increases the incentive for double-spending attacks. Ordinarily, a miner attempting a double-spend risks losing the associated block rewards if their blocks are not ultimately included on the main chain. Under \textsc{FiberPool}’s strong mining fairness, these risks are diminished.

One approach to mitigate the above impacts is to weight the PoW of shares based on their distance from the main chain. For example, one could invalidate shares whose parent blocks are not included in the main chain. Such weighting has been adopted in networks prone to frequent blockchain forks, such as Ethereum and Monero’s \textsc{P2Pool}\cite{monero_p2pool}\cite{ethereum}.

\section{Conclusion}
We proposed a new decentralized mining pool, \textsc{FiberPool}, designed to overcome several challenges encountered by existing decentralized mining pools: scalability, early-stage security, share verification costs, and budget imbalance. \textsc{FiberPool} operates through three interconnected components—a smart contract on the main chain, a storage chain, and a child chain. Additionally, we showed the mining fairness, budget balance, reward stability, and incentive compatibility in \textsc{FiberPool}.

\bibliography{hoge} 
\bibliographystyle{IEEEtran.bst}

\end{document}